\documentclass[twoside,english,final]{elsarticle}
\usepackage[T1]{fontenc}
\usepackage[latin9]{inputenc}
\usepackage{amsthm}
\usepackage{amsmath}

\makeatletter
\theoremstyle{plain}
\newtheorem{thm}{\protect\theoremname}
\theoremstyle{definition}
\newtheorem{defn}[thm]{\protect\definitionname}
\theoremstyle{plain}
\newtheorem{prop}[thm]{\protect\propositionname}
\ifx\proof\undefined\
\newenvironment{proof}[1][\protect\proofname]{\par
\normalfont\topsep6\p@\@plus6\p@\relax
\trivlist
\itemindent\parindent
\item[\hskip\labelsep
\scshape
#1]\ignorespaces
}{%
\endtrivlist\@endpefalse
}
\fi

\usepackage{times}
\@ifundefined{definecolor}
 {\usepackage{color}}{}
\usepackage{graphicx, amsmath}
\usepackage{xspace}\usepackage{array}
\definecolor{webgreen}{rgb}{0,.5,0}
\definecolor{webbrown}{rgb}{.6,0,0}
\definecolor{webyellow}{rgb}{0.98,0.92,0.73}
\definecolor{webgray}{rgb}{.753,.753,.753}
\definecolor{webblue}{rgb}{0,0,.8}

\graphicspath{{./graphics/}{.}{./figs/}}

\makeatother

\usepackage{babel}
\providecommand{\definitionname}{Definition}
\providecommand{\proofname}{Proof}
\providecommand{\propositionname}{Proposition}
\providecommand{\theoremname}{Theorem}

\begin{document}

\title{Escort entropies and divergences and related canonical distribution}

\author{J.-F. Bercher}

\ead{jf.bercher@esiee.fr}

\address{Université Paris-Est, LIGM, UMR CNRS 8049, ESIEE-Paris\\
 5 bd Descartes, 77454 Marne la Vallée Cedex 2, France}
\begin{abstract}
We discuss two families of two-parameter entropies and divergences,
derived from the standard Rényi and Tsallis entropies and divergences.
These divergences and entropies are found as divergences or entropies
of escort distributions. Exploiting the nonnegativity of the divergences,
we derive the expression of the canonical distribution associated
to the new entropies and a observable given as an escort-mean value.
We show that this canonical distribution extends, and smoothly connects,
the results obtained in nonextensive thermodynamics for the standard
and generalized mean value constraints. \end{abstract}

\maketitle
\date{August 2011}

\section{Introduction}

Rényi and Tsallis entropies extend the standard Shannon-Boltzmann
entropy, enabling to build generalized thermostatistics, that include
the standard one as a special case. The thermodynamics derived from
Tsallis entropy, the nonextensive thermodynamics, has received a high
attention and there is a wide variety of applications where experiments,
numerical results and analytical derivations fairly agree with the
new formalisms \citep{tsallis_introduction_2009}. Some physical applications
of the generalized entropies, --\,\,including statistics of cosmic
rays, defect turbulence, optical lattices, systems with long-range
interactions, superstatistical systems, etc., can be found in the
recent review \citep{beck_generalised_2009} and references therein.
In the extended thermodynamics, it has been found particularly useful
to use generalized moments \citep{abe_necessity_2005,tsallis_escort_2009}.
These moments are computed with respect to a deformed version of the
density at hands, which is called its escort distribution \citep{chhabra_direct_1989,beck_thermodynamics_1993}.
Actually, several type of constraints \citep{tsallis_role_1998,martinez_tsallis_2000,ferri_equivalence_2005}
have been used in order to derive the canonical distributions: a first
type of constraints is expressed as standard linear mean values while
a second type of constraints is given as generalized escort mean values.
In the two cases, the related canonical distributions are expressed
as a $q$-Gaussian distribution, but with two opposite exponents,
and both solutions reduce to a standard Gaussian distribution in the
$q=1$ case. These $q$-Gaussian distributions can exhibit a power-law
behavior, with a remarkable agreement with experimental data, see
for instance \citep{tsallis_introduction_2009,hilhorst_noteq-gaussians_2007,vignat_isdetection_2009},
and references therein. These distributions are also analytical solutions
of actual physical problems, e.g. \citep{lutz_anomalous_2003,schwaemmle_q-gaussians_2008,ohara_information_2010}.
There has been numerous discussions regarding the choice of a `correct'
form of the constraints, either as a standard mean or an escort-average,
and on the connections between the solutions and associated thermodynamics.
In particular, dualities and equivalences between the two settings
have been described. 

This is precisely the context of the present Letter, where we propose
a simple  connection between these different formulations and between
the related canonical distributions. More specifically, we suggest
a simple way to combine the originally distinct concepts of entropy
and escort distributions into a single two-parameter ($a,\lambda$)-entropy.
Then, we propose to look at an associated extended maximum entropy
problem. This approach includes the two aforementioned formulations
as special cases. Exploiting the nonnegativity of the associated divergence,
we derive the expression of the canonical distribution for an observable
given as an escort-mean value. We show that this canonical distribution
extends, and smoothly connects, the results obtained in nonextensive
thermodynamics for the standard and generalized mean value constraints.

We begin by recalling the context and definitions in section \ref{sec:Context-and-definitions}.
Then, the combined escort divergences and entropies are introduced
and discussed in section \ref{sec:The-a-lambda-divergences}. The
related maximum entropy problem and its solution are described in
section \ref{sec:The-a-lambda-maximum}. Finally, in section \ref{sec:The-case-of},
we illustrate the results in the case of a two-level system.

\section{\label{sec:Context-and-definitions}Context and definitions}

\subsection{Main definitions}

Let us recall that if $f$ and $g$ are two probability densities
defined with respect to a common measure $\mu,$ then for a parameter
$q>0,$ called the entropic index, the Tsallis information divergence
is defined by 
\[
D_{q}^{(T)}(f||g)=\frac{1}{q-1}\left(\int f(x)^{q}g(x)^{1-q}\text{d\ensuremath{\mu}(}x)-1\right),
\]
and similarly, the Rényi information divergence is defined by
\begin{equation}
D_{q}^{(R)}(f||g)=\frac{1}{q-1}\log\int f(x)^{q}g(x)^{1-q}\text{d\ensuremath{\mu}(}x),\label{eq:DefR=0000E9nyiDivergence}
\end{equation}
provided, in both cases, that the integral is finite. By l'Hospital's
rule, both Tsallis and Rényi information divergences reduce to the
Kullback-Leibler information in the case $q=1$. Associated to these
divergences, we obtain the entropies when $g(x)$ is uniform with
respect to $\mu$. For $q\neq1,$ 
\begin{equation}
S_{q}[f]=\frac{1}{1-q}\left(\int f(x)^{q}\text{d}\mu(x)-1\right),\label{eq:DefTsallisEntropy}
\end{equation}
is the Tsallis entropy, and 
\begin{equation}
H_{q}[f]=\frac{1}{1-q}\log\int f(x)^{q}\text{d}\mu(x),\label{eq:DefRenyiEntropy}
\end{equation}
is the Rényi entropy. In both cases, we have $S_{1}[f]=H_{1}[f]=-\int f\log f\,\text{d}\mu(x)$,
the Shannon entropy, for $q=1.$ For the Lebesgue measure, the definitions
reduces to the standard differential entropies. Finally, in the discrete
case, the continuous sum is replaced by a discrete one which extends
on a subset $\mathcal{D}$ of integers, and $\mu$ is a counting measure,
usually taken as uniform. We will also denote 
\begin{equation}
M_{q}[f]=\int f(x)^{q}\text{d\ensuremath{\mu}(}x)\label{eq:InfoGeneratingFunctionDef}
\end{equation}
the integral of the density raised to the power $q.$ 

The escort distributions used in nonextensive statistics are defined
as follows. If $f(x)$ is an univariate probability density, then
its escort $f_{a}(x)$ of order $a$, $a\geq0,$ is given by
\begin{equation}
f_{a}(x)=\frac{f(x)^{a}}{\int f(x)^{a}\mathrm{d\mu(}x)},\label{eq:escort_f-1-1}
\end{equation}
provided that $M_{a}[f]=\int f(x)^{a}\text{d\ensuremath{\mu}(}x)$
is finite. Given that $f_{a}(x)$ is the escort of $f(x)$, we see
that $f(x)$ is itself the escort of order $1/a$ of $f_{a}(x)$.
Accordingly, the (absolute) generalized $a$-moment of order $p$
is defined by
\begin{equation}
m_{p,a}[f]=\int|x|^{p}f_{a}(x)\text{d}x=\frac{\int|x|^{p}f(x)^{a}\text{d\ensuremath{\mu}(}x)}{\int f(x)^{a}\text{d}\mu(x)}.\label{eq:DefGeneralizedMoment}
\end{equation}
Of course, standard moments are recovered in the case $a=1.$

In the context of nonextensive statistical mechanics, the parameter
$a$ is taken equal to the parameter $q$ of the $q$-entropy. We
shall also point out that the `deformed' information measure like
the Rényi entropy (\ref{eq:DefRenyiEntropy}) and the escort distribution
(\ref{eq:escort_f-1-1}) are originally two distinct concepts, as
indicated here by the different notations $a$ and $q$. There is
an interesting discussion on this point in \citep{pennini_semiclassical_2007}.
It is a contribution of the present paper to indicate that it is still
possible to adopt two different values for these parameters.

\subsection{The maximum entropy problems}

In nonextensive thermostatistics, the maximum $q$-entropy distributions,
which are identified as $q$-gaussian distributions, can be obtained
using a constraint expressed as a classical mean value; but it has
also been found adequate to use constraints given as escort-mean values,
with the same index $q$ as in the entropy. This is discussed in \citep{tsallis_role_1998}
or in \citep{abe_necessity_2005,tsallis_escort_2009}. In the first
case, the maximum entropy problem stands as the variational problem
\begin{equation}
H_{q}^{(1)}(m)=\max_{f}\left\{ H_{q}[f]:\, m_{p,1}[f]=m\,\text{ and }m_{0,1}[f]=1\right\} \label{eq:MaxEntPbStandardC}
\end{equation}
which consists in finding a distribution with maximum Rényi entropy
on the set of all probability distributions with a fixed moment of
order $p$: $m_{p,1}[f]=m.$ The value of the maximum entropy obtained
for a given moment is denoted $H_{q}(m)$ \textendash{} the use of
the square brackets and parenthesis distinguishes between the functions
of the state and the functions of the observable. Similarly, for a
constraint given as an escort-mean value, we have
\begin{equation}
H_{q}^{(q)}(m)=\max_{f}\left\{ H_{q}[f]:\, m_{p,q}[f]=m\,\text{ and }m_{0,1}[f]=1\right\} .\label{eq:MaxEntPbGeneralizedC}
\end{equation}
Note that although the problems above have been written for the Rényi
entropy, the very same problems can be written for Tsallis entropy;
actually the optimum distributions are identical and the optimum values
of both maximum entropy problems are easily related. 

In the case of the standard mean constraint, the density that achieves
the maximum entropy is given by the generalized $q$-Gaussian
\begin{equation}
f(x)=\frac{1}{Z(\beta)}\left(1-\left(q-1\right)\beta|x|^{p}\right)_{+}^{\frac{1}{q-1}}\label{eq:qGaussStandardC}
\end{equation}
while in the case of the escort mean constraint, the solution has
an exponent with opposite sign
\begin{equation}
f(x)=\frac{1}{Z(\beta)}\left(1-\left(1-q\right)\beta|x|^{p}\right)_{+}^{\frac{1}{1-q}}.\label{eq:qGaussGeneralizedC}
\end{equation}
In both cases, we used the notation $\left(x\right)_{+}=\mbox{max}\left\{ x,0\right\} ,$
and $\beta$ is a positive parameter chosen so as to satisfy the constraint.
The limit case $q\rightarrow1$ gives $f(x)\propto\exp\left(-\beta|x|^{p}\right).$
The partition functions $Z(\beta)$ used to normalize $f(x)$ as a
density will be given below.

\subsection{Dualities}

There has been a great deal of discussions concerning the relationships
between the two problems (\ref{eq:MaxEntPbStandardC}) and (\ref{eq:MaxEntPbGeneralizedC})
with standard and escort mean value constraints, and between their
solutions (\ref{eq:qGaussStandardC}) and (\ref{eq:qGaussGeneralizedC}).
Dualities between the solutions have been introduced: the $q\rightarrow2-q$
transformation \citep{baldovin_nonextensive_2004,wada_connections_2005}
enables to obtain the density (\ref{eq:qGaussGeneralizedC}) from
the maximization of the entropy $S_{2-q}[f]$ subject to standard
constraints rather than the maximization of $S_{q}[f]$ subject to
escort generalized constraints. The $q\rightarrow1/q$ duality rests
on the duality between the original density and its escort: as already
noted, if $f_{q}$ is the escort of order $q$ of $f$, then $f$
is itself the escort of order $1/q$ of $f_{q}.$ This has been mentioned
in \citep{tsallis_role_1998}, the related duality discussed in \citep{Raggio1999b}
and in \citep{naudts_dual_2002}. It is very easy to obtain the $q\rightarrow1/q$
duality as the consequence of the symmetry above for escort distributions:
it suffices to note that for the Rényi entropy, we \emph{always} have
the equality $H_{\frac{1}{q}}[f_{q}]=H_{q}[f].$ Therefore, noting
that we also have $m_{p,q}[f]=m_{p,1}[f_{q}],$ we see that the problem
\begin{equation}
H_{q}^{(q)}(m)=\max_{f}\left\{ H_{q}[f]:\, m_{p,q}[f]=m\,\text{ and }m_{0,1}[f]=1\right\} ,
\end{equation}
is exactly equivalent to 
\begin{equation}
H_{q}^{(q)}(m)=\max_{f}\left\{ H_{\frac{1}{q}}[f_{q}]:\, m_{p,1}[f_{q}]=m\,\text{ and }m_{0,1}[f]=1\right\} =H_{\frac{1}{q}}^{(1)}(m)
\end{equation}
where the maximum entropy is obtained for the density (\ref{eq:qGaussGeneralizedC}).
Therefore, it is always possible to swap the standard and generalized
constraints, provided that entropic index $q$ is changed into $1/q.$
We shall also mention that these dualities have been linked to combinatorial
considerations in \citep{suyari_multiplicative_2008}.

\section{\label{sec:The-a-lambda-divergences}The ($a,\lambda$) divergences
and entropies}

In complement to the dualities mentioned above, we will show that
there is a continuum of $q$-Gaussians, solutions of an extended maximum
entropy problem, that smoothly connects the problems (\ref{eq:MaxEntPbStandardC})
and (\ref{eq:MaxEntPbGeneralizedC}) and their solutions (\ref{eq:qGaussStandardC})
and (\ref{eq:qGaussGeneralizedC}). The basic idea is to mix the concepts
of $q$-entropies and escort distributions into a single quantity.
This leads us to a simple extension of the Rényi (Tsallis) information
divergence and entropy to a two parameters case. Interestingly, the
generalized ($a,\lambda$)-Rényi information divergence that emerges
in (\ref{eq:GeneralizedRenyiDivergence}) has been mentioned \citep[eq. (49)]{cichocki_generalized_2011}
as a possible extension of a new class of ($\alpha,\beta$)-divergences,
generalizing a family of Gamma-divergences recently introduced in
\citep{fujisawa_robust_2008}. 
\begin{defn}
Given two parameters $a$ and $\lambda>0$ and two densities $f$
and $g,$ we call ($a,\lambda$)-Rényi information divergence the
standard Rényi information divergence with index $q=a/\lambda$ between
the escort distributions like (\ref{eq:escort_f-1-1}), of order $\lambda,$
$f_{\lambda}$ and $g_{\lambda},$ associated to $f$ and $g:$ 
\begin{equation}
D_{a,\lambda}^{(R)}(f||g)=D_{\frac{a}{\lambda}}^{(R)}(f_{\lambda}||g_{\lambda}).\label{eq:GeneralizedRenyiDivergenceA}
\end{equation}
In the developed form, this gives 
\begin{equation}
D_{a,\lambda}^{(R)}(f||g)=\frac{1}{a-\lambda}\log\frac{\left[\int\ensuremath{f(x)^{a}g(x)^{\lambda-a}\text{d\ensuremath{\mu}(}x)}\right]^{\lambda}}{\left[\int\ensuremath{f(x)^{\lambda}\text{d}\mu(x)}\right]^{a}\,\left[\int\ensuremath{g(x)^{\lambda}\text{d}\mu(x)}\right]^{\lambda-a}},\label{eq:GeneralizedRenyiDivergence}
\end{equation}
provided that the different integrals exist. 

Similarly, we call ($a,\lambda$)-Tsallis information divergence the
quantity
\begin{equation}
D_{a,\lambda}^{(T)}(f||g)=\frac{1}{a-\lambda}\left(\frac{\left[\int\ensuremath{f(x)^{a}g(x)^{\lambda-a}\text{d\ensuremath{\mu}(}x)}\right]^{\lambda}}{\left[\int\ensuremath{f(x)^{\lambda}\text{d}\mu(x)}\right]^{a}\,\left[\int\ensuremath{g(x)^{\lambda}\text{d}\mu(x)}\right]^{\lambda-a}}-1\right)\label{eq:GeneralizedTsallisDivergence}
\end{equation}
provided that the different integrals exist. 
\end{defn}
These definitions include the Rényi (Tsallis) information divergence
$D_{a}^{(.)}(f||g)$ in the case $\lambda=1$. In the case $\lambda=a$,
again by l'Hospital's rule, (\ref{eq:GeneralizedRenyiDivergence})
becomes $D_{\lambda,\lambda}^{(.)}(f||g)=D(f_{\lambda}||g_{\lambda})$,
the Kullback-Leibler divergence between the escort distributions of
order $\lambda,$ $f_{\lambda}$ and $g_{\lambda},$ associated to
$f$ and $g.$ Finally, in the case $a=\lambda=1$, the generalized
Rényi (Tsallis) information divergence reduces to the standard Kullback-Leibler
divergence. 

Associated to this ($a,\lambda$)-Rényi information divergences, we
can also define a two parameters Rényi (Tsallis) entropy.
\begin{defn}
Given $a$ and $\lambda>0$ and a density $f$, we call ($a,\lambda$)-Rényi
entropy the entropy of order $q=a/\lambda$ of the escort $f_{\lambda}$:
\begin{equation}
H_{a,\lambda}[f]=H_{\frac{a}{\lambda}}[f_{\lambda}]=\frac{1}{\lambda-a}\log\left[\int\ensuremath{f(x)^{a}\text{d\ensuremath{\mu}(}x)}\right]^{\lambda}\left[\int\ensuremath{f(x)^{\lambda}\text{d\ensuremath{\mu}(}x)}\right]^{-a},\label{eq:GeneralizedRenyiEntropy}
\end{equation}
provided that the different integrals exist. Similarly, the ($a,\lambda$)-Tsallis
entropy is given by
\begin{equation}
S_{a,\lambda}[f]=S_{\frac{a}{\lambda}}[f_{\lambda}]=\frac{1}{\lambda-a}\left(\left[\int\ensuremath{f(x)^{a}\text{d\ensuremath{\mu}(}x)}\right]^{\lambda}\left[\int\ensuremath{f(x)^{\lambda}\text{d\ensuremath{\mu}(}x)}\right]^{-a}-1\right),\label{eq:GeneralizedTsallisEntropy}
\end{equation}

\end{defn}
 Again, for $\lambda=1$ we get the standard Rényi (Tsallis) entropy
of order $a$, for $a=1$ we get the standard Rényi (Tsallis) entropy
of order $\lambda,$ for $a=\lambda$ it reduces to the Shannon entropy
of the escort distribution, and to the standard Shannon entropy in
the case $a=\lambda=1.$ 

In the following, we will need the fact that the divergence $D_{a,\lambda}^{(.)}(f||g)$
is always nonnegative: 
\begin{prop}
if $f$ and $g$ are two densities such that the involved integrals
are finite, with $a$ and $\lambda>0$, then $D_{a,\lambda}^{(.)}(f||g)\geq0,$
with equality if and only if $f=g.$\end{prop}
\begin{proof}
Obviously, once we have realized that $D_{a,\lambda}^{(.)}(f||g)$
in (\ref{eq:GeneralizedRenyiDivergence}) is the Rényi (Tsallis) divergence
between escort distributions, this is a direct consequence of the
nonnegativity of the Rényi (Tsallis) divergence, which follows from
Jensen inequality. Alternatively, this could also be derived from
the Hölder inequality. 
\end{proof}

\section{\label{sec:The-a-lambda-maximum}The ($a,\lambda$) maximum entropy
problem}

With the previous definitions at hands, we can now consider the following
extended maximum entropy problem:

\begin{equation}
H_{a,\lambda}(m)=\max_{f}\left\{ H_{a,\lambda}[f]:\, m_{p,a}[f]=m\,\text{ and }m_{0,1}[f]=1\right\} \label{eq:Generalized_a-lambda_problem}
\end{equation}
which includes the previous problems as particular cases. The usual
procedure for handling such variational problem is the technique of
Lagrangian multipliers, e.g. \citep{tsallis_role_1998,martinez_tsallis_2000}.
However, even though the objective functional is strictly concave
(here for $a\leq\lambda$), the constraints set is not convex and
the uniqueness and nature of the maximum can not be guaranteed. The
situation is still more involved in the continuous setting considered
here which requires the results from calculus of variations \citep{gelfand_calculus_2000,giaquinta_calculus_1996}.
Therefore, we propose to derive the solution as consequence of the
nonnegativity of $D_{a,\lambda}(f||g)$, without recourse to the technique
of Lagrange multipliers. We obtain that the density $f(x)$ that achieves
the maximum in the right hand side of (\ref{eq:Generalized_a-lambda_problem})
is a generalized Gaussian, as given by the following proposition.
\begin{prop}
The density that achieves the maximum of the $(a,\lambda)$-Rényi
or Tsallis entropy, subject to the escort mean constraint $m_{p,a}[f]=m$,
is the generalized Gaussian \textup{
\begin{equation}
G_{\beta}(x)=\begin{cases}
\frac{1}{Z(\beta)}\left(1-\left(\lambda-a\right)\beta|x|^{p}\right)_{+}^{\frac{1}{\lambda-a}} & \text{for }\lambda\not=a\\
\frac{1}{Z(\beta)}\exp\left(-\beta|x|^{p}\right) & \text{if }\lambda=a
\end{cases}\text{ }\label{eq:qgauss_general-1}
\end{equation}
}where $\beta$ is a positive parameter. The partition function $Z(\beta)$
is given by 
\begin{alignat}{1}
 & Z(\beta)=\frac{2}{p}\left(\beta\right)^{-\frac{1}{p}}\times\nonumber \\
 & \begin{cases}
\left(a-\lambda\right)^{-\frac{1}{p}}B\left(\frac{1}{p},-\frac{1}{\left(\lambda-a\right)}-\frac{1}{p}\right) & \!\!\!\!\mbox{for }a>\lambda>a-p\\
\left(\lambda-a\right)^{-\frac{1}{p}}B\left(\frac{1}{p},\frac{1}{\left(\lambda-a\right)}+1\right) & \!\!\!\!\mbox{for }\lambda>a\\
\Gamma\left(\frac{1}{p}\right) & \!\!\!\!\text{and if }\lambda=a
\end{cases}\label{eq:GenPartitionFunction-1}
\end{alignat}
where $B(x,y)$ is the Beta function. Finally, the $a$-moment of
order $p$ is given by 
\begin{equation}
m_{p,a}[G_{\beta}]=\begin{cases}
\frac{1}{\beta p(\lambda-a)\left(\frac{a}{\lambda-a}+\frac{1}{p}+1\right)} & \text{for }\lambda-a\neq0\\
\frac{1}{\beta pa} & \text{for }\lambda=a.
\end{cases}\label{eq:generalized_nu_moment_p=00003Dalpha}
\end{equation}

\end{prop}
For $\lambda>a$, the density has a compact support, while for $\lambda\leq a$
it is defined on the whole real axis and behaves as a power distribution
for $|x|\rightarrow\infty.$ 
\begin{proof}
The expressions of the partition function and of the $a$-moment are
obtained by direct calculation. For the main result, let us consider
the generalized Gaussian (\ref{eq:qgauss_general-1}), and denote
$A(\beta)=1/Z(\beta)$ the inverse of the partition function. We directly
have
\begin{gather}
\int\!\ensuremath{f^{a}G_{\beta}^{\lambda-a}\text{d}\mu(x)=A(\beta)^{\lambda-a}M_{a}[f]\times\int\!\ensuremath{\left(1-(\lambda-a)\beta|x|^{p}\right)\negthinspace_{_{+}}}\frac{f^{a}}{M_{a}[f]}\text{d}\mu(x)}\!\label{eq:treize}\\
\leq A(\beta)^{\lambda-a}\!\left(1-(\lambda-a)\beta\, m_{p,a}[f]\right)M_{a}[f],\label{eq:quatorze}
\end{gather}
where we have exhibited the escort distribution $f_{a}$, where $m_{p,a}[f]$
denotes the generalized $a$-moment, and where the inequality in (\ref{eq:quatorze})
results from the fact that for $\lambda>a,$ the support of $\left(1-(\lambda-a)\beta|x|^{p}\right)_{+}$
can be smaller than the support of $f.$ From (\ref{eq:quatorze})
we also immediately get with $f=G_{\beta}$ that
\begin{equation}
M_{\lambda}[G_{\beta}]=A(\beta)^{\lambda-a}\left(1-(\lambda-a)\beta\, m_{a,p}[G_{\beta}]\right)M_{a}[G_{\beta}].\label{eq:seize}
\end{equation}
Therefore, for all distributions $f$ with a given $a$-moment $m_{p,a}[f]=m$
and for $\beta$ such that the generalized Gaussian has the same moment
$m_{p,a}[G_{\beta}]=m$, then the combination of (\ref{eq:quatorze})
and (\ref{eq:seize}) gives 
\begin{align*}
\int\ensuremath{f^{a}G_{\beta}^{\lambda-a}\text{d}\mu\,\leq} & \frac{M_{\lambda}[G_{\beta}]}{M_{a}[G_{\beta}]}M_{a}[f].
\end{align*}
Finally, the $(a,\lambda)$-Rényi divergence can be expressed as 
\begin{alignat*}{1}
D_{a,\lambda}^{(R)}(f||G_{\beta})=\log & \left(\frac{\left[\int\ensuremath{f^{a}G_{\beta}^{\lambda-a}\text{d}\mu(x)}\right]^{\lambda}}{M_{\lambda}[f]^{a}\, M_{\lambda}[G_{\beta}]^{\lambda-a}}\right)^{\frac{1}{a-\lambda}}\leq\log\left(\frac{M_{a}[G_{\beta}]^{-\lambda}M_{\lambda}[G_{\beta}]^{a}}{M_{a}[f]^{-\lambda}M_{\lambda}[f]^{a}}\right)^{\frac{1}{a-\lambda}}=H_{a,\lambda}(G_{\beta})-H_{a,\lambda}(f),
\end{alignat*}
by the definitions (\ref{eq:GeneralizedRenyiDivergence}) and (\ref{eq:GeneralizedRenyiEntropy})
of the divergence and entropy. By the nonnegativity of the divergence,
we obtain that 
\[
H_{a,\lambda}(G_{\beta})\geq H_{a,\lambda}(f)
\]
for all distributions $f$ with a given $a$-moment $m_{p,a}[f]=m_{p,a}[G_{\beta}]=m$,
and with equality iff $f=G_{\beta}$. Using the $(a,\lambda)$-Tsallis
divergence rather than Rényi's, the same result also directly follows
for the ($a,\lambda)$-Tsallis entropy. 
\end{proof}
From this result, we see that the problem (\ref{eq:Generalized_a-lambda_problem})
and its solution indeed interpolates between the maximum $q$-entropy
problem with standard constraint (\ref{eq:MaxEntPbStandardC}) and
with generalized $q$-moment constraint (\ref{eq:MaxEntPbGeneralizedC}). 
\begin{itemize}
\item For $a=q$ and $\lambda=1,$ the ($q,1$)-Rényi entropy is the standard
Rényi entropy with index $q$, whose maximum subject to a $q$-moment
constraint is attained for the generalized Gaussian (\ref{eq:qGaussGeneralizedC})
with exponent $1/(1-q).$ 
\item For $a=1$ and $\lambda=q,$ the ($1,q$)-Rényi entropy is the Rényi
entropy with index $q,$ whose maximum subject to a standard moment
constraint is the generalized Gaussian (\ref{eq:qGaussStandardC})
with exponent $1/(q-1).$
\end{itemize}
It shall be also mentioned that we still get a $q\rightarrow1/q$
duality result in this setting. From the definition (\ref{eq:GeneralizedRenyiEntropy}),
we always have the identity $H_{a,\lambda}[f]=H_{\frac{a}{\lambda}}[f_{\lambda}]=H_{\frac{\lambda}{a}}[f_{a}]$.
In the other hand, the generalized $a$-moment $m_{p,a}[f]$ can be
considered as the standard moment of density $f_{a}$, $m_{p,a}[f]=m_{p,1}[f_{a}]$
but also as the escort moment of order $a/\lambda$ of the density
$f_{\lambda},$ i.e. $m_{p,a}[f]=m_{p,\frac{a}{\lambda}}[f_{\lambda}]$.
Therefore, the problem (\ref{eq:Generalized_a-lambda_problem}) can
be recasted in the two equivalent forms

\begin{equation}
H_{a,\lambda}(m)=\max_{f}\left\{ H_{\frac{\lambda}{a}}[f_{a}]:\, m_{p,1}[f_{a}]=m\,\text{ and }m_{0,1}[f]=1\right\} \label{eq:Generalized_a-lambda_problem-2}
\end{equation}
which is a problem with respect to the density $f_{a}$ under a standard
mean constraint, or 

\begin{equation}
H_{a,\lambda}(m)=\max_{f}\left\{ H_{\frac{a}{\lambda}}[f_{\lambda}]:\, m_{p,\mbox{\ensuremath{\frac{a}{\lambda}}}}[f_{\lambda}]=m\,\text{ and }m_{0,1}[f]=1\right\} \label{eq:Generalized_a-lambda_problem-3}
\end{equation}
which is a problem with respect to $f_{\lambda}$with a generalized
moment constraint.

\section{\label{sec:The-case-of}The case of a two-level system}

Finally, we close this paper with the example of a two-level system,
with eigenenergies 0 and 1. Although the general ($a,\lambda)$ maximum
entropy problem (\ref{eq:Generalized_a-lambda_problem}) usually do
not lead to closed-form expressions for the entropies $H_{a,\lambda}(m)$,
see e.g. \citep{bercher_entropy_2008}, we still can get an explicit
solution for the simple two-level system. For this system, the measure
at hands charges the two levels, with $\mu=\delta_{0}+\delta_{1}$
where $\delta_{x}$ denotes the Dirac mass at $x$. Then, the maximum
entropy distribution is the discrete distribution with a probability
$p$ for the excited state and ($1-p$) for the ground state, with
\[
p=\frac{\left(1+\beta\right)^{\frac{1}{\lambda-a}}}{1+\left(1+\beta\right)^{\frac{1}{\lambda-a}}}.
\]
 In turn, the ($a,\lambda)$-entropy is then given by 
\[
H_{a,\lambda}[G_{\beta}]=\frac{1}{\lambda-a}\log\left[p^{a}+(1-p)^{a}\right]^{\lambda}\left[p^{\lambda}+(1-p)^{\lambda}\right]^{-a},
\]
and the internal energy, $m,$ which is computed with respect to the
escort distribution of order $a$ is simply 
\[
m=\frac{p^{a}}{p^{a}+(1-p)^{a}}.
\]
With a little algebra, it is possible to express $p$ as a function
of $m$:
\[
p=\frac{m^{\frac{1}{a}}}{m^{\frac{1}{a}}+\left(1-m\right)^{\frac{1}{a}}},
\]
and then we can get the expression of the ($a,\lambda)$-entropy:
\[
H_{a,\lambda}(m)=\frac{a}{a-\lambda}\log\left[m^{\frac{\lambda}{a}}+(1-m)^{\frac{\lambda}{a}}\right].
\]
From this expression, it is easy to derive an expression of the inverse
temperature as the derivative with respect to $m,$ or the expression
of the related heat capacity. We also see that the subsequent entropy
only depends on the ratio $\lambda/a.$ Finally, for $\lambda=a,$
l'Hospital's rule gives
\[
H_{a,a}(m)=-m\log\left[m\right]-\left(1-m\right)\log\left(1-m\right),
\]
which is the Fermi-Dirac entropy.

\section{Conclusions}

In this Letter, we have suggested a possible extension of the standard
$q$-divergences and entropies by considering a combination of the
concepts of $q$-divergences and of escort distributions. This leads
to two-parameter information measures that recover the classical ones
as special cases, and coincide with some recently proposed measures
\citep{cichocki_generalized_2011}. Further work should examine the
general properties of these information measures. We have introduced
a general maximum entropy problem and derived the expression of the
canonical distribution maximizing the ($a,\lambda)$-entropy under
a generalized $a$-moment constraint. This was obtained as a consequence
of the nonnegativity of the related information divergence, without
recourse to the technique of Lagrange multipliers. This canonical
distribution is a two-parameter version of the generalized $q$-Gaussian
distribution, extending the versatility of this distribution. It includes
the standard solutions of the nonextensive thermostatistics as particular
cases. Actually, the proposed generalized maximum entropy problem
includes the standard approaches of nonextensive thermostatistics
and provides a continuum of problems and solutions between them. In
addition, this approach suggests that it is possible to adopt different
indexes for the $q$-entropy and for the escort-mean constraint.



\begin{thebibliography}{10}
\expandafter\ifx\csname url\endcsname\relax
  \def\url#1{\texttt{#1}}\fi
\expandafter\ifx\csname urlprefix\endcsname\relax\def\urlprefix{URL }\fi

\bibitem{tsallis_introduction_2009}
C.~{T}sallis, Introduction to Nonextensive Statistical Mechanics, 1st Edition,
  Springer, 2009.

\bibitem{beck_generalised_2009}
C.~Beck, Contemp. Physics 50 (2009) 495.

\bibitem{abe_necessity_2005}
S.~Abe, G.~B. Bagci, Phys. Rev. E 71 (2005) 016139.

\bibitem{tsallis_escort_2009}
C.~{T}sallis, A.~R. Plastino, R.~F. {Alvarez-Estrada}, J. Math. Phys. 50 (2009)
  043303.

\bibitem{chhabra_direct_1989}
A.~Chhabra, R.~V. Jensen, Phys. Rev. Lett. 62 (1989) 1327.

\bibitem{beck_thermodynamics_1993}
C.~Beck, F.~Schloegl, Thermodynamics of Chaotic Systems, Cambridge University
  Press, 1993.

\bibitem{tsallis_role_1998}
C.~{T}sallis, R.~S. Mendes, A.~R. Plastino, Physica A 261 (1998) 534.

\bibitem{martinez_tsallis_2000}
S.~Mart\'{i}nez, F.~Nicol\'{a}s, F.~Pennini, A.~Plastino, Physica A 286 (2000)
  489.

\bibitem{ferri_equivalence_2005}
G.~L. Ferri, S.~Martinez, A.~Plastino, J. Stat. Mech. (2005) P04009.

\bibitem{hilhorst_noteq-gaussians_2007}
H.~J. Hilhorst, G.~Schehr, J. Stat. Mech. Theor. Exp. 2007 (2007) P06003.

\bibitem{vignat_isdetection_2009}
C.~Vignat, A.~Plastino, Physica A 388 (2009) 601.

\bibitem{lutz_anomalous_2003}
E.~Lutz, Phys. Rev. A 67 (2003) 051402.

\bibitem{schwaemmle_q-gaussians_2008}
V.~Schw\"ammle, F.~D. Nobre, C.~{T}sallis, European Phys. J. B 66 (2008) 537.

\bibitem{ohara_information_2010}
A.~Ohara, T.~Wada, J. Phys. A Math. Gen. 43 (2010) 035002.

\bibitem{pennini_semiclassical_2007}
F.~Pennini, A.~Plastino, G.~Ferri, Physica A 383 (2007) 782.

\bibitem{baldovin_nonextensive_2004}
F.~Baldovin, A.~Robledo, Phys. Rev. E 69 (2004) 045202.

\bibitem{wada_connections_2005}
T.~Wada, A.~Scarfone, Phys. Lett. A 335 (2005) 351.

\bibitem{Raggio1999b}
G.~A. Raggio, cond-mat/9908207, (1999).

\bibitem{naudts_dual_2002}
J.~Naudts, Chaos Solitons Fractals 13 (2002) 445.

\bibitem{suyari_multiplicative_2008}
H.~Suyari, T.~Wada, Physica A 387 (2008) 71.

\bibitem{cichocki_generalized_2011}
A.~Cichocki, S.~Cruces, S.~ichi Amari, Entropy 13 (2011) 134.

\bibitem{fujisawa_robust_2008}
H.~Fujisawa, S.~Eguchi, J. Multivar. Anal. 99 (2008) 2053.

\bibitem{gelfand_calculus_2000}
I.~M. Gelfand, S.~V. Fomin, Calculus of Variations, Dover Publications, 2000.

\bibitem{giaquinta_calculus_1996}
M.~Giaquinta, S.~Hildebrandt, Calculus of Variations: The {L}agrangian
  formalism, Springer, 1996.

\bibitem{bercher_entropy_2008}
J.~F. Bercher, Inf. Sci. 178 (2008) 2489.

\end{thebibliography}

\end{document}